\documentclass[draftcls, 11pt, onecolumn]{ieeetran}
\IEEEoverridecommandlockouts
\usepackage{amssymb, amsmath, graphicx, paralist, subfigure, cite, algorithm, algorithmic}
\newtheorem{lem}{Lemma}
\newtheorem{theorem}{Theorem}
\newtheorem{defn}{Definition}
\newtheorem{prop}{Proposition}
\newtheorem{cor}{Corollary}
\newtheorem{rem}{Remark}

\def\mb{\mathbf}
\def\mc{\mathcal}

\begin{document}
\title{\bf \Large On the genericity properties in networked estimation:\\ Topology design and sensor placement}
\author{Mohammadreza Doostmohammadian and Usman A. Khan
\thanks{
Department of Electrical and Computer Engineering, Tufts University, {\texttt{\{mrd,khan\}@ece.tufts.edu}}. }}

\maketitle

\begin{abstract}
In this paper, we consider networked estimation of linear, discrete-time dynamical systems monitored by a network of agents. In order to minimize the power requirement at the (possibly, battery-operated) agents, we require that the agents can exchange information with their neighbors only \emph{once per dynamical system time-step}; in contrast to consensus-based estimation where the agents exchange information until they reach a consensus. It can be verified that with this restriction on information exchange, measurement fusion alone results in an unbounded estimation error at every such agent that does not have an observable set of measurements in its neighborhood. To over come this challenge, state-estimate fusion has been proposed to recover the system observability. However, we show that adding state-estimate fusion may not recover observability when the system matrix is structured-rank ($S$-rank) deficient.

In this context, we characterize the state-estimate fusion and measurement fusion under both full $S$-rank and $S$-rank deficient system matrices. The main results of this paper are the following. \emph{Firstly}, we show that when the system matrix has full $S$-rank, state-estimate fusion alone (with no measurement fusion) can recover the observability. Subsequently, we characterize the minimal topology for inter-agent communication required for a stable networked estimator. \emph{Secondly}, we provide methodologies to recover (networked) estimator observability when the system matrix is $S$-rank deficient. In particular, we introduce a novel agent classification based on their local measurements and identify the agents that are crucial for stable estimation error. We then provide topology modifications and sensor placement techniques to recover observability in the $S$-rank deficient scenario. Finally, we provide an iterative method to compute the local estimator gain at each agent that results into a stable estimation error once the observability is ensured using the aforementioned techniques.


\end{abstract}

\section{Introduction}\label{intro}

Estimation of dynamical systems with observations distributed
among a network of agents is an important field of research, where the idea is to assign a group of agents to monitor a certain system or phenomenon of interest. Agents are distributed in the sense that each agent can only measure some of the states of a dynamical system, referred to as local measurements. For example, a group of sensors spread geographically over a large region to monitor daily temperature evolution. The measurement data and dynamical models are further corrupted by noise and disturbances. The objective is to enable each agent to make an unbiased decision on the global state relying only on its own measurement and the measurements from its immediate neighbors. Such a scheme is often referred to as \emph{networked estimation} where the term network implies that the information is restricted on a sparse network.

Networked estimation is preferable to a wide range of applications as it is scalable and further requires less communication load at each individual agent, in contrast to the centralized case where each agent may require repeated long-distance communication to a central location. Applications of networked estimation include social networks \cite{asu} to learn global beliefs based on partial understanding of the state of the society, market, politics, etc., monitoring physical processes and environmental spatio-temporal fields \cite{kriged, MiadCons}, state estimation in power systems \cite{camsap11, XieKar11, ZhangReMeDySE}, and multi-agent systems such as collaborative target tracking and flocking of mobile robots \cite{flock}, see for example Fig.~\ref{drones}.
\begin{figure}
\centering
\includegraphics[width=1.5in]{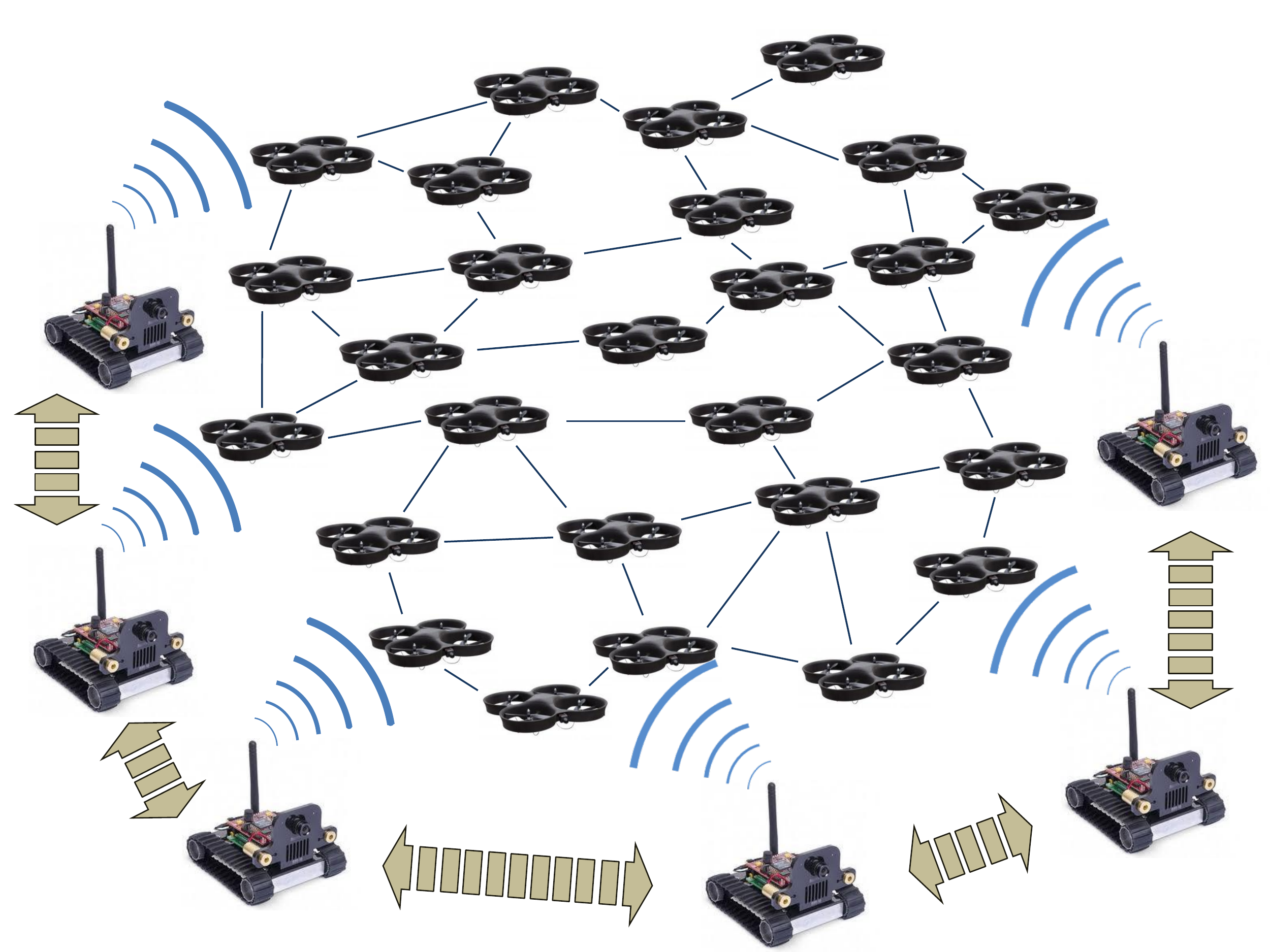}
\caption{A group of mobile ground-robots are tracking a large flock of quad-copters. Every ground-robot is monitoring part of the flock, and then shares its partial data with other ground-robots through a communication network.}
\label{drones}
\end{figure}

A variety of solutions exists for networked estimation starting from the earlier work in \cite{hashem:88, Mutambara_book} and references therein on parallel Kalman filter architectures for all-to-all connected networks, to more recent diffusion-based schemes via least mean square implementation, such as in Kalman filtering and smoothing \cite{sayed-kf} and distributed binary detection \cite{sayed11}. Meanwhile, incremental adaptive distributed strategies can be found in \cite{sayed-increment,incrementbertsekas} along with distributed moving horizon estimation \cite{farina-mhe} to minimize estimation error variance for constrained problems. State estimators based on low-cost single-bit data transmission is proposed in \cite{sign} with binary sign of innovations (sign of difference of measurement and estimated value). Information theoretic approach based on consensus over the Kullback-Leibler average of Gaussian PDFs is exploited in \cite{battistelli2011information}. The literature can also be classified into static and dynamic estimation. In static estimation \cite{kriged, sayed-kf, sayed11, sayed-increment, sayedtu12, sayedchen11}, the target state to be estimated does not change over time, while dynamic estimation \cite{camsap11, hashem:88, farina-mhe, sign, battistelli2011information, olfati:05, zamp:07, usman_tsp:07, Msechu:08, usman_acc:11} takes into account the time-evolution of the system\footnote{As stated in \cite{sayedtu12}, diffusion algorithms can be extended for non-stationary (dynamic) tracking when the target is not moving too fast, i.e. its state is relatively stationary over a period such that the algorithm can converge.}.

Consensus-based strategies have recently found a lot of interest in the context of sparsely-connected networks, where the main focus is to reduce the uncertainty of individual estimates by averaging on collaborative data. Early work in \cite{olfati:05, zamp:07, usman_tsp:07, Msechu:08, usman_acc:11} considers a two time-scale method where consensus is implemented at a time-scale different than the system dynamics. These results require that a consensus is reached within every two time-steps of the system dynamics, and is thus, challenged with a large number (infinite, in general) of consensus iterations between every two steps of the dynamics. To elaborate this, consider Fig.~\ref{cps_ts_1}(a), where a large number ($\rightarrow\infty$) of data fusion iterations are implemented between every two successive time-steps, $k$ and $k+1$, of the dynamics.
\begin{figure}
\centering
\includegraphics[width=2.5in]{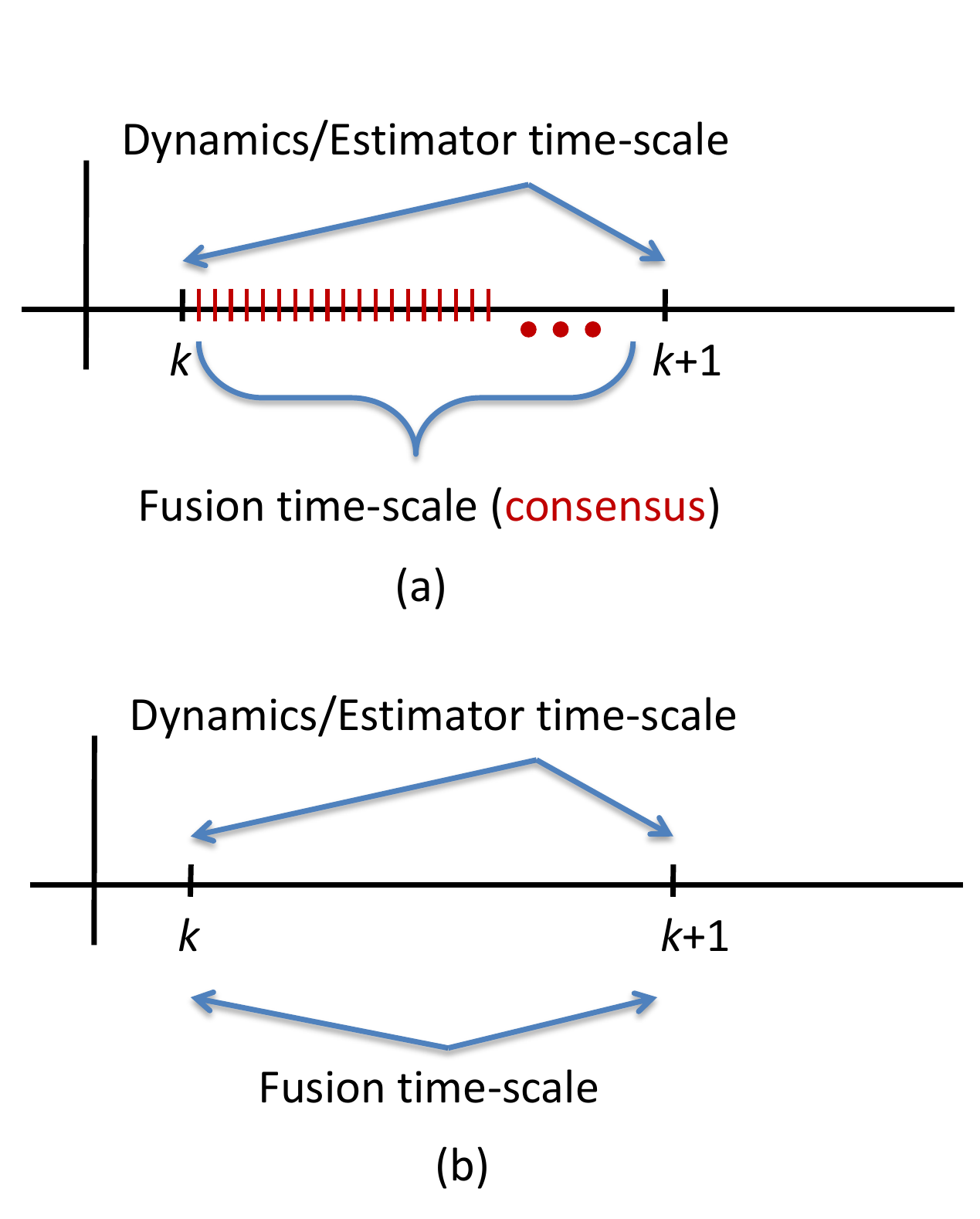}
\caption{(a) The traditional two-time scale consensus-based approach. (b) single time-scale approach.}
\label{cps_ts_1}
\end{figure}
This approach requires communication over a much faster rate than the sampling of the dynamics, and thus, in general, becomes practically infeasible when the underlying system is operating under power constraints and has restricted communication and computation budgets.

In contrast to the two time-scale approach to distributed estimation, recently References \cite{asu, sayedtu12, sayedchen11,battistelli2011information, usman_acc:11, usman_cdc:10,kar-moura-ramanan-IT-2008} studied the behavior of networked estimators when the communication time-scale is the same as the time-scale of the dynamics, as shown in Fig.~\ref{cps_ts_1}(b). This method is practically feasible for real-time applications and computationally efficient as compared to the two time-scale approach. A preliminary study on this single time-scale estimator is carried out in \cite{usman_cdc:10}, where it is shown that a particular linear networked estimator has a bounded estimation error if the two-norm of the system matrix is less than the network tracking capacity--a function of the communication network and observation models. Notice that in the two time-scale method, the communication network becomes irrelevant due to more information exchanges among the individuals (the information in a sparsely connected graph is equivalent to the information in a fully connected graph when a large number of information exchanges are carried out). Therefore, the performance and properties of the underlying estimator depends only on the data fusion principles among the agents. However, in the single time-scale scenario of Fig.~\ref{cps_ts_1} (b), the underlying agent network remains sparse and an arbitrary communication network may not suffice to make the networked estimation error stable (e.g., see \cite{usman_cdc:10,usman_cdc:11, asilomar11}).

In this context, the key problem is to design the structure of the inter-agent communication according to the underling fusion rules in order to recover the observability of the networked estimator. In this paper, we use a variant of the Networked Kalman-type Estimator (NKE) protocol, initially introduced in \cite{usman_cdc:10}. The main contribution is to determine the communication network among the agents to recover the observability of the underlying estimation protocol, given that each agent may not be locally (in its neighborhood) observable. We study the observability with a structural point of view \cite{rein_book,liu-nature,sauter:09,woude:03,boukhobza-recovery,equitable-egerstedt,egerstedt-nature} in the sense that we explore the generic properties of the system. The generic properties are applicable to any choice of system parameters as long as the sparsity structure (zeros and non-zeros) is not violated. The generic approach is helpful when the underlying system parameters may change depending on the system operating point (linearization of non-linear dynamics) and is further significant in communication network design as the approach is independent of the exact value of the weights chosen for data fusion. Moreover, this implies that for smooth non-linear systems with fixed structure Jacobian matrix, similar analysis of the networked observability can be applied.

Comparing with other work in the literature, we consider single time-scale estimation, as opposed to the multi time-scale estimation in \cite{olfati:05, zamp:07, usman_tsp:07, Msechu:08} and the \textit{vanishing time-step} algorithms proposed in \cite{ram10, binachi11}. Unlike \cite{hierarchy-giannakis,hierarchy-egerstedt}, we do not impose an agent hierarchy (i.e., we assume the processing/communication duties at all agents are the same). Avoiding agent hierarchy increases the reliability of node/link failure. We further do not require the communication network to be (strongly) connected \cite{sayed-kf,sayed11, usman_cdc:11,battistelli2011information,sayedtu12,zamp:07} or for it to include a cyclic path \cite{sayed-increment,sayedchen11}. Our goal is to design the network with minimal communication. Specifically, we use methodologies that are independent of exact system values and rely only on the \emph{structure} of the underlying system. This leads to a robust estimator design where the analysis is not algebraic, as in the conventional Grammian or PBH observability tests, but graph-theoretic \cite{sauter:09,woude:03}.

We now describe the rest of the paper. Section~\ref{pre} provides preliminary material on basic dynamical system estimation and structured systems theory, whereas Section~\ref{pfff} presents our problem formulation. Section~\ref{as_cl} enlists our assumptions and describes a novel agent classification method. Section~\ref{main} covers the main results of this paper on state and output fusion, whereas local gain design is explored in Section~\ref{K}. We provide an illustrative example and simulations in Section~\ref{example}, and finally, Section~\ref{conc} concludes the paper.

\section{Background and Preliminaries}\label{pre}
We consider the system model to be a discrete-time linear dynamical system:
\begin{eqnarray}\label{sys1}
\mb{x}_{k+1} &=& A\mb{x}_k + \mb{v}_k,
\end{eqnarray}
where~$\mb{x}_k\in\mathbb{R}^n$ is the state vector, ~$A=\{a_{ij}\}\in\mathbb{R}^{n\times n}$ is the system matrix, and~$\mb{v}_k\sim\mathcal{N}(0,V)$ is the system noise. We note that the main emphasis of this paper is not on modeling but on structural characteristics of the underlying system and the results we present hold for any phenomenon following Eq.~\eqref{sys1}. We assume that the dynamical system is monitored by a network of~$N$ agents such that each sensor~$i$ has the following observation model:
\begin{eqnarray}\label{sys2_ag}
\mb{y}^i_k &=& C_i\mb{x}_k + \mb{r}^i_{k},
\end{eqnarray}
where~$\mb{y}^i_k\in\mathbb{R}^{p_i}$ is the output vector at agent~$i$,~$\mb{r}^i_k\sim\mathcal{N}(0,R_i)$ is the output noise, and~$C_i$ is the output matrix at agent~$i$. With this notation, we can write the global observation model as
\begin{eqnarray}
\mb{y}_k = C\mb{x}_k + \mb{r}_k,
\end{eqnarray}
where
\begin{eqnarray}\label{Cdef}
\mb{y}_k =
\left[
\begin{array}{c}
\mb{y}_k^1\\
\vdots\\
\mb{y}_k^N
\end{array}
\right],~~
C =
\left[
\begin{array}{c}
C_1\\
\vdots\\
C_N
\end{array}
\right],~~
\mb{r}_k =
\left[
\begin{array}{c}
\mb{r}_k^1\\
\vdots\\
\mb{r}_k^N
\end{array}
\right],
\end{eqnarray}
$\mb{r}_k\sim\mathcal{N}(0,R)$ is the global observation noise
with~$R = \mbox{blockdiag}[R_1,\ldots,R_N]$, and~$C=\{c_{ij}\}$ is
the global output matrix.

\subsection{Centralized estimator}\label{cntr}
Let~$\widehat{\mb{x}}^c_{k|k}$ be the centralized Kalman estimator
\cite{kalman:61} at time~$k$ given all the observations,
$\mb{y}_k$, up to time~$k$. It can be shown that the error, $\widehat{\mb{e}}^c_{k|k} = \mb{x}_k - \widehat{\mb{x}}^c_{k|k}$, in this estimator is given by 
\begin{eqnarray}\label{ge}
\widehat{\mb{e}}^c_{k|k} = (A - K_cCA)\widehat{\mb{e}}^c_{k-1|k-1} + \eta_k,
\end{eqnarray}
where~$K_c$ is the centralized Kalman gain and the vector~$\eta_k$
collects the remaining terms that are independent of
$\widehat{\mb{e}}^c_{k-1|k-1}$. It is well known that the
centralized Kalman error,~$\widehat{\mb{e}}^c_{k|k}$ is stable if and only if all the unstable modes (eigenvalues) of the system are observable. For the ease of explanation, we assume that there are no stable unobservable nodes. In other words, detectability and observability are equivalent throughout this paper.

In the traditional sense of $n$-step $(A,C)$-observability, the observability Gramian is given by
\begin{eqnarray}\label{pbh}
\mathcal{O} = \left[C^T ~~ A^TC^T ~ ... ~ (A^{n-1})^TC^T\right].
\end{eqnarray}
Algebraic tests for observability check the Gramian, $\mathcal{O}$, for being full-rank or the matrix $\mathcal{O}^T\mathcal{O}$ for being invertible. An alternative method is the PBH (Popov-Belevitch-Hautus) observability test \cite{hautus}, which requires the matrix,~$[A^T-sI ~~ C^T]$, to be full-rank for all $s$. The matrix~$[A^T-sI]$ is full rank for all values of $s$ other than the eigenvalues of $A$ and, therefore, the PBH test is needed to be checked \textit{only} for these values.

Note that, both these algebraic methods rely on the knowledge of exact values of each element in the matrices $A$ and $C$.  However, in many dynamical systems, only the sparsity (zero and non-zero pattern) of these matrices may remain fixed while the non-zero elements are subject to change. For example, when the elements of the concerned matrices depend on certain parameters or operating points. Hence, these conventional methodologies fail to check for observability in such cases and graph-theoretic techniques are to be employed. We introduce such graph-based methods below. 

\subsection{Graph notations}
Let~$X=\{x_1,\ldots,x_n\}$ denote the
state set, and let~$Y=\{y_1,\ldots,y_p\}$ 
denote the output set. We define the
\emph{system digraph} as~$G_A = (V,E)$, where~$V=X\cup Y$ is the
vertex set, and~$E$ is the edge set containing directed edges,
$(v_1,v_2)\in E$, of the form~$v_1\rightarrow v_2$ with
$v_1,v_2\in V$. The edge set~$E$ is defined as~$E_A\cup E_C$,
where~$E_A=\{(x_j,x_i)~|~a_{ij}\neq0\}$ and
$E_C=\{(x_j,y_i)~|~c_{ij}\neq0\}$. A \emph{path} of length~$\ell$
from~$v_1\in V$ to~$v_{\ell}\in V$ is such that there exists a
sequence of vertices,~$ v_1,v_2,\ldots,v_{\ell}$ with each subsequent edge,
$(v_1,v_2),(v_2,v_3),\ldots,(v_{\ell-1},v_\ell)\in E$. Here~$v_1$
is the \emph{begin-vertex} of the path and~$v_\ell$ is its
\emph{end-vertex}. Here, we assume that each vertex contained in a
path occurs only once (\emph{simple} path). A path is said to be
\emph{$Y$-topped} if it ends at a vertex in~$Y$. A digraph is called strongly connected if there exist a directed path from each vertex to every other vertex in the digraph. In a not strongly connected digraph, define \emph{Strongly Connected Components (SCC)} as its maximal strongly connected partitions or sub-graphs.
A \emph{cycle} is a simple path where the begin and end vertices are the same. Since the nodes in $Y$ have no outgoing link, nodes included in a cycle all belong to $X$. As an example, consider Fig.~\ref{figAC} which shows the system digraph of a dynamical system with~$n=7$ states (encircled) and~$N=3$ measurements (or agents) denoted by squares.
\begin{figure}
\centering
{\includegraphics[width=2.2 in]{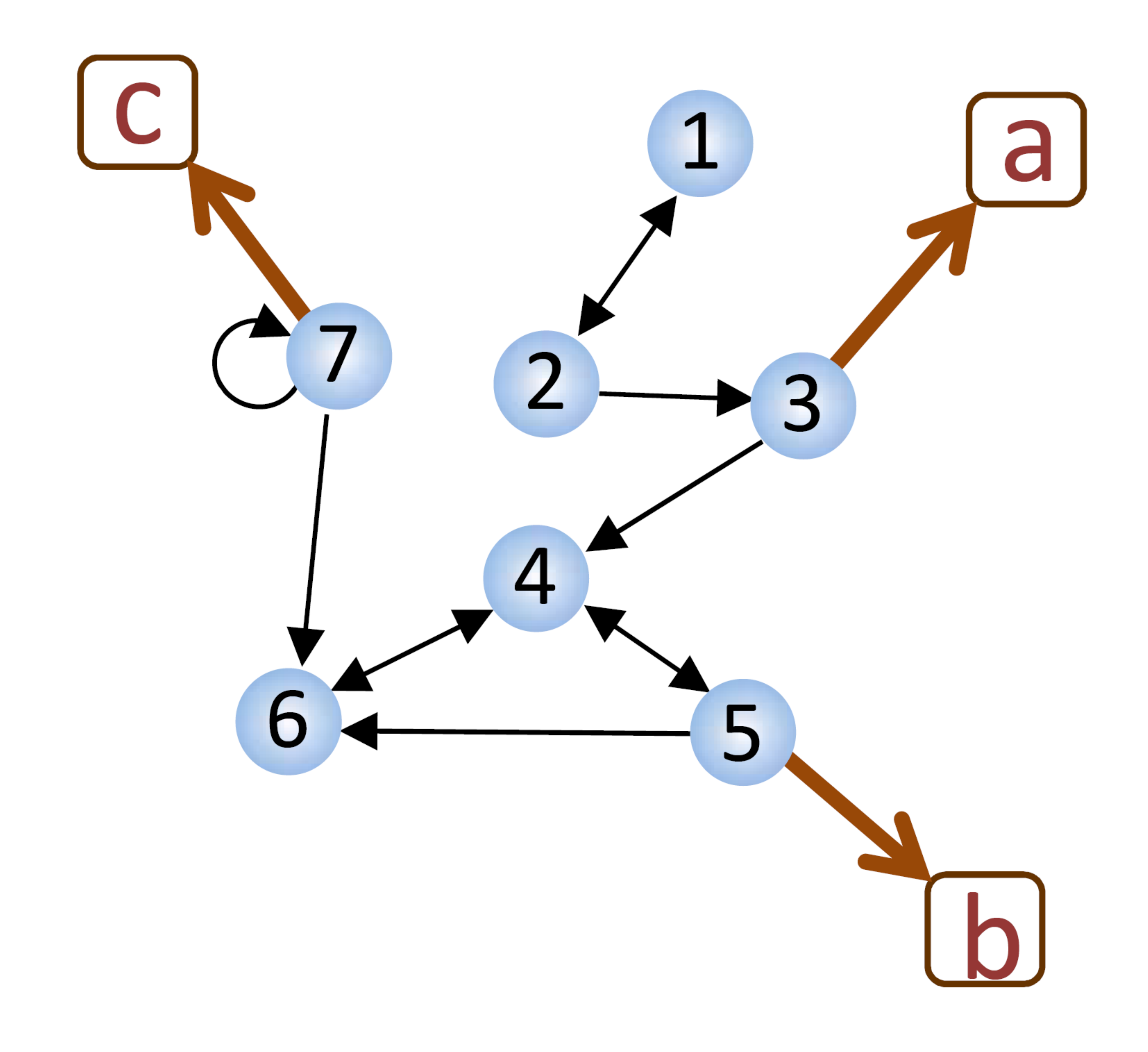}}
\caption{This digraph is an example of $(A,C)$~observable system based on the Theorem~\ref{th1}.}
\label{figAC}
\end{figure}
We now introduce some new concepts on SCCs over state vertices. These will be helpful in describing our results. 
\begin{defn}[Parent SCC]\label{parent}
A state SCC, is called a \textit{parent SCC}, if it has no outgoing link to any state vertex not belonging to itself.
\end{defn}
\begin{defn}[Child SCC]
Any SCC that is not a parent SCC is a \textit{child SCC}.
\end{defn}
Notice that the set of disjoint state SCCs in system matrix~$A$ can be explicitly characterized as either a parent or a child. As an example, the SCC containing vertices $\{4,5,6\}$ in Fig.~\ref{figAC} is a parent SCC, since there is no outgoing edges from its states to other states  $\{1,2,3,7\}$ not included in it. Furthermore, $\{1,2\}$ and $\{7\}$ are child SCCs. More details on parent/child SCCs and efficient algorithms for computing SCCs in a digraph can be found in~\cite{asilomar11} and ~\cite{algorithm}, respectively. We now use the concepts from this section to formally introduce structured systems theory and generic properties.

\subsection{Structured systems theory}
Structural analysis deals with system properties that do not
depend on the numerical values of the parameters but only on the underlying structure (zeros and non-zeros) of the system \cite{rein_book,liu-nature,sauter:09,woude:03,boukhobza-recovery,equitable-egerstedt,egerstedt-nature}. It turns out that if
a structural property is true for one \textit{admissible}
choice of non-zero elements as free parameters it is true for
\textit{almost all} choices of non-zero elements and, therefore, is called a \emph{generic} property of the
system~\cite{wang:73}. Furthermore, it can be shown that those
particular (non-admissible) choices for which the generic property
does not hold lie on some algebraic variety with zero Lebesgue
measure, for more details see
\cite{wang:73, wang:74}.

\begin{defn}[$S$-rank]\label{s-rank}
The structural rank (also called generic rank) is the maximum rank for all numerical values of the non-zero entries of the matrix~$A$. It is, in fact, an upper-bound on the numerical rank of~$A$.
\end{defn}

The $S$-rank as a generic property holds for almost all choices of
nonzero parameters of the matrix,~$A$. It is equal to the cardinality of the maximum matching associated to the \textit{bipartite graph} associated to the matrix,~$A$. In the algebraic sense, this is the maximum number of non-zero elements in distinct rows and columns of the matrix,~$A$ \cite{harary}. Details on the generic rank implication in graph theoretic sense and algorithms on maximum matching can be found in~\cite{sauter:09, algorithm}. Among other generic properties, controllability/observability are of interest in the context of this paper, see~\cite{lin,liu-nature,equitable-egerstedt,sherayas:08,sauter:09,egerstedt-nature} for details. We extend the following theorem from the generic controllability results in~\cite{woude:03}.

\begin{theorem}\label{th1}
A dynamical system is generically observable if and only if in the
system digraph:
\begin{enumerate}[(i)]
\item Every state is the begin-node of a path that ends in an output (termed as a~$Y$-topped path);
\item There exist a disjoint union of~$Y$-topped paths and cycles that cover all the state vertices.
\end{enumerate}
\end{theorem}

The following lemma is from \cite{rein_book}.

\begin{lem}\label{lemth1}
The condition (ii) in Theorem~\ref{th1} on the generic observability of~$(A_{n \times n},C_{m \times n})$ is equivalent to the following:
\begin{equation}
S\mbox{-rank}\left(
\left[
\begin{array}{c}
A\\
C
\end{array}
\right]
\right) = n.
\end{equation}
\end{lem}
The proof of Theorem~\ref{th1} for generic controllability and Lemma~\ref{lemth1} is given in ~\cite{rein_book,woude:03}, where other equivalent graph-theoretic conditions to generic controllability (observability) are also defined that
we omit here. As an example, consider the system shown in Fig.~\ref{figAC}. It can be verified that each state is a begin-vertex of a~$Y$-topped path, and~$\{7\}$,~$\{4,5,6\}$,~$\{1,2,3,a\}$ constitute a disjoint union of cycles and $Y$-topped paths that cover all the state vertices in $X$. Thus, satisfying both conditions in Theorem~\ref{th1} and the system in Fig.~\ref{figAC} is observable for almost all choices of non-zero elements.

\subsection{Corresponding graphs}
In this paper, we deal with two different graph representations: \textit{system digraph},~$G_A$, representing states of the dynamic system ~\eqref{sys1} and~\eqref{sys2_ag}, and digraph~$G_W$ defining the agent\textit{communication network}. Let~$G_W = (V_W,E_W)$, where~$V_W=\{1,\ldots,N\}$ is the vertex set consisting of $N$ agents,
$E_W=\{(i,j)~|~i\leftarrow j\}$ is the edge set, and
$\mathcal{D}_i=\{i\}\cup\{j~|~(i,j)\in E\}$ denote the extended
neighborhood of agent~$i$. Notice that, unlike many works in the literature we do not constrain $G_W$ to be undirected. In fact, no assumption on the topology is considered here, as designing~$G_W$ is a contribution of this paper.

\emph{Example:} To shed more light on this, we give an example here. Consider the flocking motion example given in the Fig.~\ref{drones}. The position, velocity or acceleration of every quad-copters can be considered as a state of the system. having ground robots as agents, the coordination law among them, typically following dynamics \eqref{sys1} \cite{blondelflock:05}, defines the dynamical system, $A$, and system digraph, $G_A$. The system outputs/measurements are the states tracked by the ground robots, and the communication network (to be designed) among these ground-robots is $G_W$.

\section{Problem formulation}\label{pfff}
We employ a variant of the Networked Kalman-type Estimator (NKE) proposed in
\cite{usman_cdc:11, asilomar11}. Let~$\widehat{x}^i_{k|m}$ be the
state estimate of agent~$i$ at time~$k$ given the outputs up to time $m$, ($m\leq k$), from its neighboring agents,~$j\in\mathcal{D}_i$.
Each agent implements the following:
\begin{enumerate}[(i)]
\item \emph{Predictor and state fusion:}
\begin{eqnarray}\label{lp}
\widehat{\mb{x}}^i_{k|k-1} &=& \sum_{j\in\mathcal{D}_i} w_{ij}A\widehat{\mb{x}}^j_{k-1|k-1}
\end{eqnarray}
\item \emph{Estimator and output fusion:}
\begin{eqnarray}\label{le}
\widehat{\mb{x}}^i_{k|k} &=&\widehat{\mb{x}}^i_{k|k-1} + K_k^i \sum_{j\in D_i}C_j^T (y^j_k-C_j\widehat{\mb{x}}^i_{k|k-1})
\end{eqnarray}
\end{enumerate}
where~$W=\{w_{ij}\}$ is the \textit{state fusion} weight matrix
such that~$w_{ij}\geq 0$ with~$\sum_{j\in\mc{D}_i}w_{ij}=1$ ($W$
is stochastic), and~$K_k^i$ is the local estimator gain at agent $i$. 
\begin{rem} \label{wrank}
Following are some useful remarks: 
\begin{inparaenum}[(i)]
\item The diagonal entries of~$W$ are all nonzero, since every agent is in its own extended neighborhood and uses its own information.
\item The first equation~\eqref{lp} is a local prediction fusion where each agent $i$ fuses the neighboring estimates from time $k-1$ and then implements a local predictor.
\item In the second equation~\eqref{le}, each agent $i$ updates its local prediction with an \textit{innovation} term. We define this innovation as the difference between the state prediction of agent~$i$ and the state measurements obtained via agents,~$j\in\mc{D}_i$. Adding this term, agent, $i$, makes its final estimate, $\widehat{\mb{x}}^i_{k|k}$, for the current time step.
\item The protocol given in Eqs.~\eqref{lp}--\eqref{le} takes place at the same time-scale as the system dynamics, see Fig.~\ref{cps_ts_1}. Notice that both Eqs.~\eqref{lp}--\eqref{le} can be combined into one equation; we give separate equations for the ease of explanation.
\end{inparaenum}
\end{rem}

Let the estimation error at agent~$i$ and time~$k$ be defined as 
\begin{eqnarray}\label{love}
\mb{e}_{k}^i = \mb{x}_{k|k} - \widehat{\mb{x}}^i_{k|k},
\end{eqnarray}
and let~$\mb{e}_{k} = [(\mb{e}_{k}^1)^T,\ldots,(\mb{e}_{k}^N)^T]^T$ be the networked estimation error derived in the following.
\begin{prop}\label{ned_pr}
Let $\mathbf{q}_k^i$ be some function of the system and measurement noise, $\mathbf{v}_k$ and $\mathbf{r}_k^i$, independent of $\mathbf{e}_{k-1}$ and let 
\begin{eqnarray}\nonumber
K_k &=& \mbox{blockdiag}[K_k^1,\ldots,K_k^N],\\\nonumber D_C &=&
\left[
\begin{array}{ccc}
\sum_{j\in\mathcal{D}_1} C_j^TC_j&&0\\
&\ddots&\\
0&&\sum_{j\in\mathcal{D}_N} C_j^TC_j
\end{array}
\right],\\
\mb{q}_k &=& [(\mb{q}_k^1)^T, \hdots , (\mb{q}_k^N)^T]^T.
\end{eqnarray}
Then we get the following networked error dynamics,
\begin{eqnarray}\label{err1}
\mb{e}_{k} = (W\otimes A - K_kD_C(W\otimes A))\mb{e}_{k-1} +
\mb{q}_k.
\end{eqnarray}
\end{prop}
The derivation requires some straightforward manipulations and is omitted here. Comparing this to Eq.~\eqref{ge}, we note that the networked estimation error, $\mathbf{e}_k$, can be stabilized if and only if, the following pair,
\begin{eqnarray}\label{netO}
(W\otimes A, D_C),
\end{eqnarray}
is observable. In other words, a gain matrix, $K_k$, exists such that
$\rho(W\otimes A-K_kD_C(W\otimes A)) < 1$ (i.e., it is a Schur matrix), if and only if $(W\otimes A, D_C)$ is observable, where $\rho(\cdot)$ denotes the spectral radius of a matrix. As it can be seen from Eq.~\eqref{netO}, the communication network, $W$, plays a major role in distributed estimation as opposed to the multiple time-scale approach where $W$ is irrelevant. The role of $W$ in observability is because of the single time-scale nature of the estimator, see Fig.~\ref{cps_ts_1}.

\begin{rem}
~
\begin{itemize}
\item The variables $D_C$ and $K_k$ are block-diagonal matrices.
\item Every block diagonal, $\sum_{j\in\mathcal{D}_i} C_j^TC_j$, in the matrix~$D_C$, can be thought of as a representation of all of the measurements in the extended neighborhood of agent~$i$.
\end{itemize}
\end{rem}
We refer to $(W\otimes A, D_C)$ as the \textit{distributed system} and
$G_{W\otimes A}$ as the graph associated with the matrix~$W\otimes A$. For better understanding of the structural relevance of the estimator in~\eqref{lp}--\eqref{le}, we first consider $W=I$ and $\overline{D_C}$ defined as follows,
\begin{eqnarray}
\overline{D_C} &=& \left[
\begin{array}{ccc}
C_1^TC_1&&\\
&\ddots&\\
&&C_N^TC_N
\end{array}
\right],
\end{eqnarray}
implying no information exchange among the agents. This distributed system, $(I\otimes A, \overline{D_C})$, can be thought of as~$N$ subsystems each of them associated to an $n \times n$ block diagonal in the matrix~$W\otimes A$, see Fig.\ref{figmatrixW} (Left).
\begin{figure}
\centering
\subfigure{\includegraphics[width=3in]{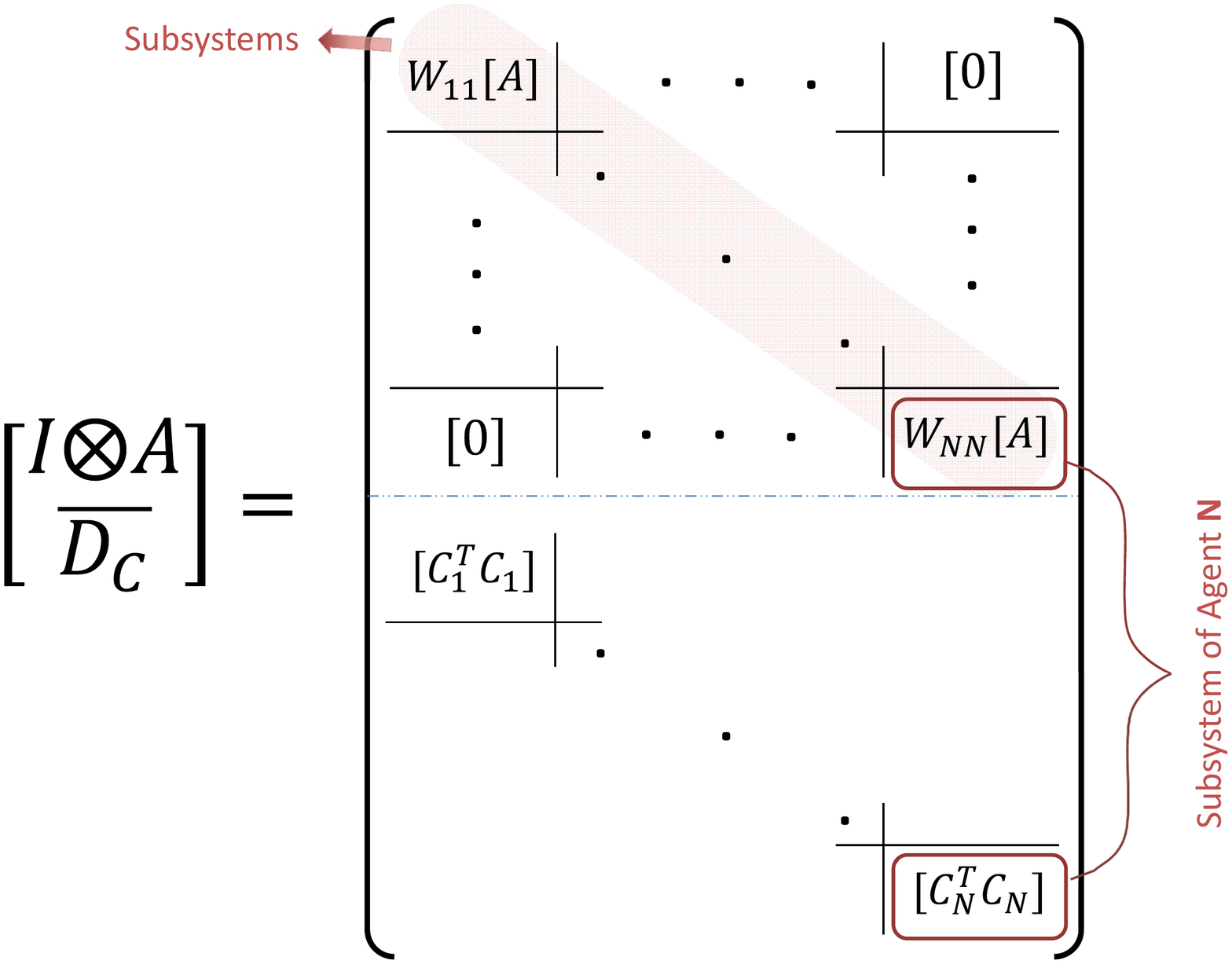}}
\subfigure{\includegraphics[width=3in]{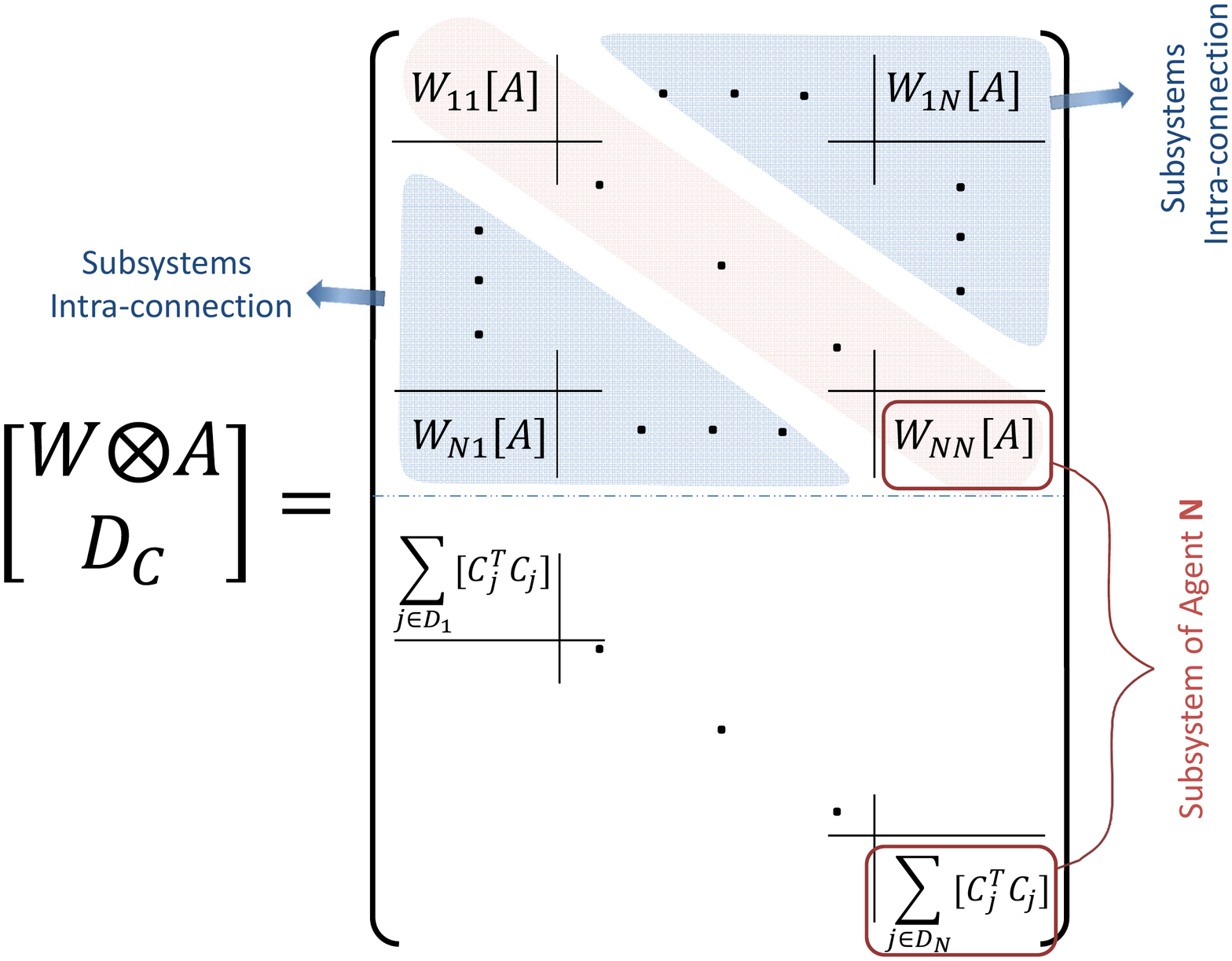}}
\caption{(Left) Matrix structure of the distributed system with no data fusion. Every block diagonal $W_{ii} \otimes A$ is a subsystem associated with the output/agent~$i$. (Right) Adding data fusion, the intra-connections among these subsystems depends on the non-diagonals $W_{ij} \otimes A$, $i \neq j$.}
\label{figmatrixW}
\end{figure}
Now consider $W$ to have some non-zero non-diagonal entries. As it is shown in Fig.\ref{figmatrixW} (Right), these entries define the \textit{inter-connections} among these subsystems.

To shed more light on this, consider the example given in Fig.~\ref{figAC} where we show a~$n=7$-state dynamical system with~$N=3$ agents,~$\{a,b,c\}$ such that agent~$a$ measures~$x_3$, agent~$b$ measures~$x_5$ and agent~$c$ measures~$x_7$. Each agent is \textit{required to} estimate the entire $n=7$ dimensional state-vector. Without any information fusion each agent only has a partial observation of the system as it is shown in Fig.~\ref{figA}.
\begin{figure*}
\centering
{\includegraphics[width=6 in]{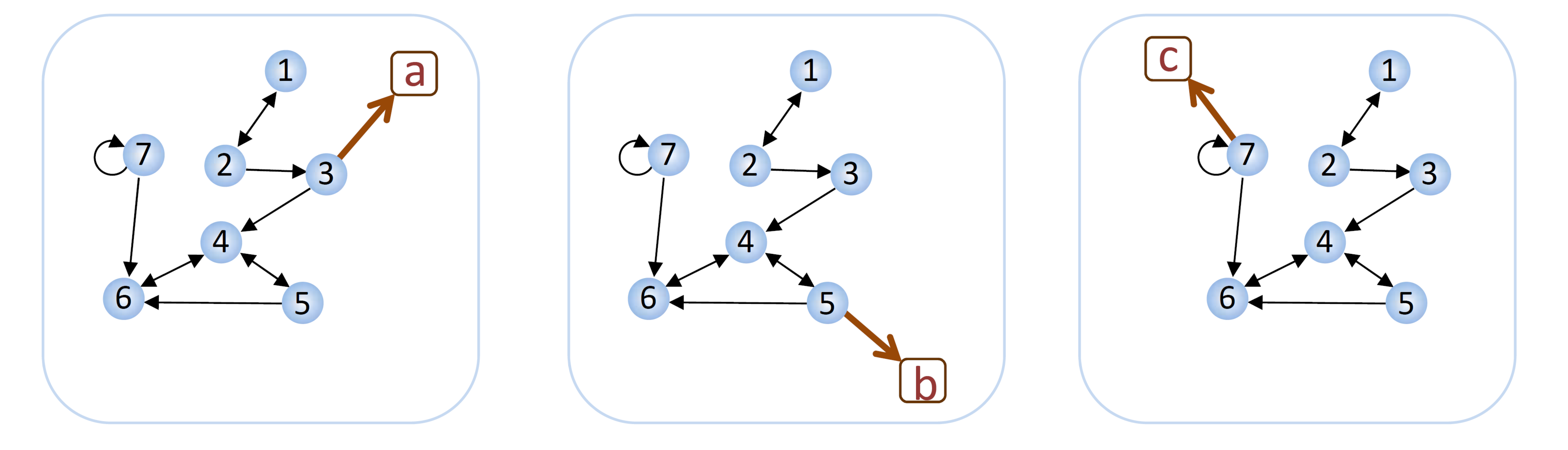}}
\caption{The graph associated with distributed system, $(I \otimes A,\overline{D_C})$, with no data fusion represented as a sub-system at each agent. According to the Theorem~\ref{th1}, each sub-system (agent) is not observable with no data fusion.}
\label{figA}
\end{figure*}
Therefore, each agent has to acquire the missing information via communicating with agents in its immediate neighborhood. However, in this illustration, no agent finds any observation in its neighborhood in addition to what it already possesses. Information sharing among the agents by applying state and output fusion provides more links among the subsystems in the distributed system digraph. This extra linking among subsystems and outputs, captured by the non-zeros in $W$ and the summation in $D_C$, has the \textit{potential} to improve the generic observability of the system. In this regard, the main objective is to define the \textit{structure} of the communication matrix~$W$ (graph~$G_W$) such that the distributed system~$(W\otimes A, D_C)$ is generically observable.

We seek a general method to make each subsystem observable. First, we describe how adding a link between two agents changes the graph structure of the distributed system. We explain this by considering the same example as in Fig.~\ref{figAC} and Fig.~\ref{figA}. In the case of output-fusion, a link between two agents, for example, from agent~$b$ to agent~$a$ ($a\leftarrow b$), implies that agent~$a$ has access to agent~$b$'s measurement, that is measurement of state~$x_5$. However, the state-fusion case is more involved. For example, adding a link from agent~$b$ to agent~$a$ implies a nonzero entry in~$W$, i.e., $w_{ab\neq0}$,  which reflects in the networked system matrix, $W\otimes A$, as adding edges \emph{to} some states in the subsystem associated to agent~$a$ \emph{from} some states in the subsystem associated to agent~$b$. This will be discussed in more detail in the next sections.

\section{Our approach}\label{as_cl}
We now enlist our assumptions and provide a novel agent classification that will help in establishing our results. 
\subsection{Assumptions}\label{as_cl_a}
In the rest of the paper, we make the following assumptions:
\begin{enumerate}[(i)]
\item The communication between the agents is stable, i.e., the the network is static;
\item For every agent, $i$, the pairs, $(A,C_i)$ or    $(A,\sum_{j\in\mathcal{D}_i}C_j^TC_j)$, are not necessarily observable;
\item The system is globally~$(A,C)$-observable, i.e., if we collect all the sensor measurements at a center then the dynamical system is observable.
\end{enumerate}

Assumption~(ii), in practice, makes the networked estimation problem more challenging and is where this work becomes significantly different from current approaches, see, for example,~\cite{flock} and references therein. Assumption (iii) is a typical assumption in distributed estimation implying the observability of centralized estimator; without this, no estimation scheme will work.

\subsection{State and Agent classification}\label{as_cl_cl}
To describe our approach, we provide a novel agent classification. Since the system is $(A,C)$-observable, (iii) in Section~\ref{as_cl_a}, we can enlist a disjoint set of cycles and $Y$-topped paths that covers all the state vertices (existence is ensured from condition (i)--Theorem~\ref{th1}). We are interested in a listing that involves the maximal cycles and we denote this set as $\mathcal{L}$. For example, from Fig.~\ref{figAC}, the disjoint set of cycles and $Y$-topped paths that covers all the state vertices includes $\{(4,6,4), (5,b), (1,2,3,a), (7,c)\}$, and $\{(4,5,6,4), (1,2,1), (7,7), (3,a)\}$, among others. However, the latter includes the maximal cycles and thus $\mathcal{L}=\{(4,5,6,4), (1,2,1), (7,7), (3,a)\}$. The following classification is with respect to~$\mathcal{L}$.
\begin{enumerate}[(i)]
\item \emph{Type-$\alpha$ agent} is an agent that appears in the $Y$-topped paths in $\mathcal{L}$. For example, agent~$a$ in Fig.~\ref{figAC}.
\item \emph{Type-$\beta$ agent} is an agent that measures a state in the parent cycles cycles in $\mathcal{L}$; a parent cycle is a cycle that does not have an outgoing link to any other state not belonging to itself. For example, agent~$b$ in Fig.~\ref{figAC}.
\item \emph{Type-$\gamma$ agent} is an agent that measure a state in the child cycles in $\mathcal{L}$; a child cycle is a cycle that is not a parent cycle. For example, agent~$c$ in Fig.~\ref{figAC}.
\end{enumerate}
The above agent classification leads to the following definition and lemma. 
\begin{defn}[Crucial observation] A crucial observation is a measurement such that removing it renders the dynamical system unobservable.
\end{defn}
\begin{lem}\label{cr_o}
The agents of Type-$\alpha$ and Type-$\beta$ make ``crucial" observations while the measurements at Type-$\gamma$ agents are not crucial.
\end{lem}
\begin{proof}
Since the Type-$\beta$ agents monitor the parent cycles and there is no outgoing link from a parent cycle to any other state outside this cycle, the states in the parent cycles can only be the begin vertices of a $Y$-topped path (in order to satisfy condition (ii) in Theorem~\ref{th1}) when any one of these states is connected to an output. Hence all the Type-$\beta$ agents make crucial observations. On the other hand, removing a Type-$\alpha$ agent violates the condition (i) in Theorem~\ref{th1} as the attached state vertex is not included in $\mathcal{L}$ anymore. Hence, Type-$\alpha$ agents are also crucial. Finally, the only location for the Type-$\gamma$ agent is monitoring a child SCC, which either has a directed path to a Type-$\alpha$ agent or to a Type-$\beta$ agent and hence is redundant. While the contribution of this state remains in $\mathcal{L}$ due to the cycle present there.
\end{proof}

For example, in Fig.~\ref{figAC}, if either agent $a$ or agent $b$ is removed, then the system becomes unobservable. It can also be verified that agent $c$ is non-crucial. Having defined types of agents, we note that the observability of the distributed system can be recovered via either~$W\otimes A$ matrix (state-fusion) or~$D_C$ (output-fusion). Here, we provide the \textit{minimal}
sufficient number of communication among the agents. Unlike our
previous works~\cite{usman_cdc:11, asilomar11}, we do not impose
any constraint on the system matrix,~$A$. Furthermore, the generic approach is further robust to uncertain systems and to \textit{linearized} approximation of nonlinear models where the structure is fixed while the values are a function of the operating point \cite{liu-nature}.

\section{Recovering observability}\label{main}
In this section, we first present some helpful results for the development of the paper and then find a general solution for $(W \otimes A, D_C)$~observability. We first discuss the role of state fusion, related to the structure of matrix~$W$, and then the role of output fusion, related to the structure of matrix~$D_C$ (see Table~\ref{tab}). The reason is to get more intuitive and separate solutions for state and output fusion; obviously, in real applications if two agents are linked together they nay share all of their information, including both their measurement and state estimates, to maximally improve their current state estimates. The results and proofs in this section are mainly graph theoretic that is a direct consequence of using the generic approach.

\begin{table}
\centering
\caption{This table shows distributed system for different fusion levels.}
\begin{tabular}{l*{2}{c}r}
Fusion level      & equivalent distributed system \\
\hline
No data fusion &  $(I \otimes A , \overline{D_C})$ \\
Only state fusion    & $(W \otimes A , \overline{D_C})$  \\
Only output fusion &  $(I \otimes A , D_C)$ \\
Both measurement and state fusion &  $(W \otimes A , D_C)$ \\
\end{tabular}

\label{tab}
\end{table}

\subsection{Results on rank genericity}
The result below follows from Theorem~\ref{th1} and Lemma~\ref{lemth1} as provided in Section~\ref{pre}.
\begin{cor}[Full~$S$-rank]\label{fullsrank}
A system matrix,~$A$, is full~$S$-rank if and only if its associated digraph has a disjoint union of cycles covering all the state vertices.
\end{cor}
Notice that non-zero diagonals of a matrix can be represented as a disjoint union of self-cycles in its associated digraph. From Corollary~\ref{fullsrank} and by Remark~\ref{wrank} (Section~\ref{pfff}) we obtain the following result.
\begin{cor}\label{rankW}
The communication matrix,~$W$, has a disjoint union of self-cycles and,
\begin{eqnarray}
S\mbox{-rank}(W)=N.
\end{eqnarray}
\end{cor}
This is always true because $W$ has non-zero diagonals, i.e., $w_{ii}\neq0,~\forall i$. Consequently, we state the following lemma for the networked system~$(W \otimes A)$.
\begin{lem}\label{rankW*A}
For the communication matrix,~$W_{N \times N}$, and system matrix,
$A_{n \times n}$, the networked system~$W\otimes A$ is
structurally full-rank if and only if~$A$ is structurally
full-rank. Mathematically,
\begin{eqnarray}
S\mbox{-rank} (W\otimes A)=N \times n \iff S\mbox{-rank}(A)=n
\end{eqnarray}
\end{lem}
\begin{proof}
Recall that for two matrices,~$W$ and~$A$,
\begin{equation}
rank(W\otimes A)=rank(W)\times rank(A)
\end{equation}
From Corollary~\ref{rankW}, we have~$rank(W)=N$ for almost all numerical values, and for any full
rank matrix~$A_{n\times n}$, we have
\begin{equation}
\max(rank(W\otimes A))=N\times n,
\end{equation}
Based on the definition of the $S$-rank, we can conclude that~$(W \otimes A)$ is generically full rank for almost all choices of numerical values. This proves the
necessity. On the other hand, if~$rank(A)<n$ for any choice of~$W$, then we have
\begin{equation}
\max(rank(W\otimes A))<N\times n,
\end{equation}
which implies that,
\begin{equation}
S\mbox{-rank}(W\otimes A)< N\times n.
\end{equation}
This proves the sufficiency.
\end{proof}

\subsection{State fusion}
We now explore Eq.~\eqref{lp} in NKE protocol and assume that there is no output fusion. In particular, we analyze the structure of the matrix~$W$ for $(W \otimes A , \overline{D_C})$~observability according to Table~\ref{tab}. First, we provide some special cases where the system matrix,~$A$, is structurally full rank. This is the case, for example, in linearization and discretization of non-linear systems where the system matrix
\emph{almost always} has non-zero diagonal entries.

\begin{lem} \label{placement}
For full~$S$-rank system matrix, $(A,C)$ is centrally observable if and only if every parent~SCC is output connected, i.e., monitored by (at least) one agent.
\end{lem}
The proof is straightforward and omitted here. Interested readers may see our previous work in \cite{asilomar11}. The following theorem establishes conditions on the communication
network,~$G_W$, over full~$S$-rank systems.
\begin{theorem} \label{fullA,W}
With a full~$S$-rank system,~$A$, the pair $(W \otimes A, \overline{D_C})$ is generically observable when for every parent-SCC in~$A$, say~$\mathcal{K}$, if agent~$i$ does not have an observation of a state in $\mathcal{K}$, then in the communication network,~$G_W$, there must be a directed path from agent~$i$ to any agent~$j$, which has a state observation in\footnote{If there is more than one agent observing SCC~$\mathcal{K}$, say agents~$j,k$, a directed path from agent $i$ to only one of them is sufficient.} $\mathcal{K}$.
\end{theorem}
\begin{proof}
The system matrix~$A$ being full $S$-rank ensures the condition~(i) in Theorem~\ref{th1}. This is because from Corollary~\ref{fullsrank}, there exists a disjoint union of cycles alone that cover all of the state vertices and the $Y$-topped paths are not needed to verify condition~(i). To satisfy condition~(ii), all state vertices in a subsystem associated to every agent, say~$i$, must be a begin vertex of a $Y$-topped path. This condition, according to Lemma~\ref{placement}, is satisfied by having every parent-SCC in $W \otimes A$ be output-connected.

Since in communication matrix~$W$ there is a path from agent $i$ to $j$, in $G_{W \otimes A}$~graph, subsystem of agent $i$ is connected to subsystem of agent $j$. Therefore, every state vertex in parent-SCC~$\mathcal{K}$ in subsystem $i$ is connected to parent-SCC~$\mathcal{K}$ in subsystem $j$ (see Fig.~\ref{figproof}).
\begin{figure}
\centering
\includegraphics[width=3 in]{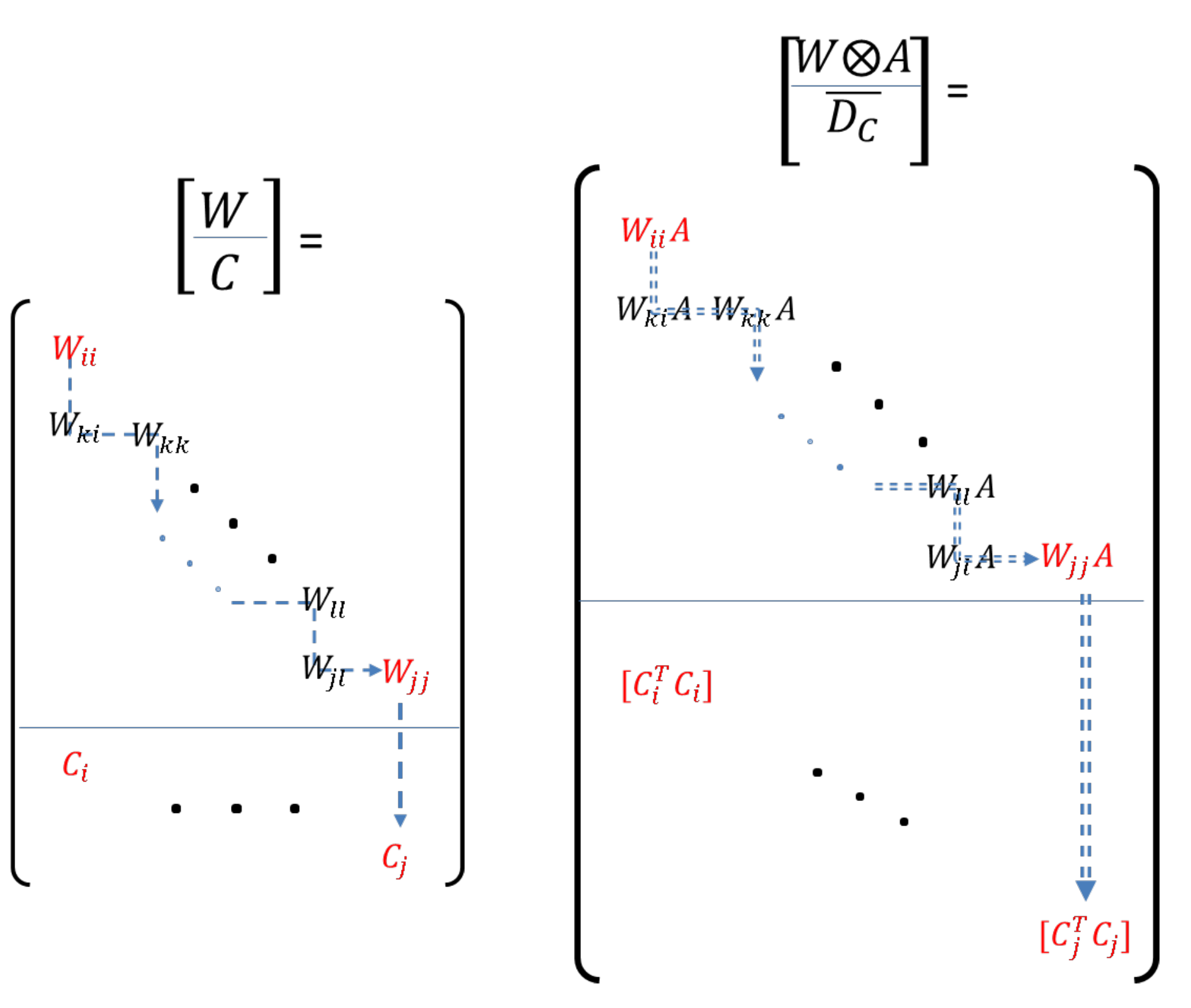}
\caption{This figure illustrates the proof of the Theorem~\ref{fullA,W} showing that: A directed path from agent~$i$ to agent~$j$ in $W$ matrix (left) implies a  directed path from states in subsystem~$i$ to subsystem~$j$ and consequently agent~$j$ in $(W \otimes A)$~matrix.}
\label{figproof}
\end{figure}
Since every state of parent-SCC~$\mathcal{K}$ in subsystem $j$ is $Y$-topped path to output $j$, every state of parent-SCC~$\mathcal{K}$ in subsystem~$i$ is also connected to output~$j$. With this for every parent-SCC~$\mathcal{K}$ in every subsystem $i$ of $G_{W\otimes A}$, all parent-SCCs are output connected and the theorem follows.
\end{proof}

For example, consider again the three-output system in Fig.~\ref{figA}. Having vertices~$\{4,5,6\}$ as parent-SCC, agent~$b$ is the Type-$\alpha$ agent. According to the above theorem any other agent without any observation in~$\{4,5,6\}$, like agent~$a$, must be connected to agent~$b$. This provides a connection from parent-SCC~$\{4,5,6\}$ in subsystem~$a$ to its counterpart SCC in subsystem~$b$ in distributed system graph~$G_{W \otimes A}$, and in turn its output connectivity.

A very important point to mention here is that for full~$S$-rank systems, there only exist Type-$\beta$ and Type-$\gamma$ agents. We prove this in the following lemma. 
\begin{lem}
If a system matrix has full $S$-rank then we only have Type-$\beta$ or Type-$\gamma$ agents.
\end{lem} 
\begin{proof}
The proof is straightforward and relies on Corollary~\ref{fullsrank}. Since $A$ is full $S$-rank, there exists a set of disjoint cycles that covers all of the state vertices. Hence, the set $\mathcal{L}$ introduced for agent classification in Section~\ref{as_cl_cl} does not include any $Y$-topped paths and thus we cannot have any Type-$\alpha$ agent.
\end{proof}
However, when the system matrix is not full~$S$-rank, we also encounter Type-$\alpha$ agents that possess crucial observations from Lemma~\ref{cr_o}. It turns out that if the system matrix is not full~$S$-rank, then even using a fully-connected communication network (complete $G_W$~graph) does not recover observability. We now provide our main result on state fusion.
\begin{theorem} \label{non-full A,W}
Assume that $(A,C_i)$ is not observable at any agent $i$. If system,~$A$, is not full~$S$-rank, then the NKE~\eqref{lp}--\eqref{le} is not observable with state fusion alone, i.e., ~$(W\otimes A,\overline{D_C})$ is not observable for any choice of the communication matrix~$W$.
\end{theorem}
\begin{proof}
Let~$i$ be an agent for which condition (i) in Theorem~\ref{th1} does not hold, i.e.,
\begin{equation}
S\mbox{-rank}\left(
\left[
\begin{array}{c}
A\\
C^T_iC_i
\end{array}
\right]
\right) <n.
\end{equation}
Such an agent always exists because: (i) based on the assumption (ii) in Section~\ref{pfff}, the entire system is not observable at any agent; and (ii) the matrix~$A$ is not full~$S$-rank. Now consider~$(W \otimes A, \overline{D_C})$ for the best-case scenario where $G_W$ is an all-to-all network.
Let~$W_i$ be the~$i$th column of the adjacency matrix~$W$. Obviously,~$W_i \otimes A$ is the~$i$th \textit{block column} of~$(W \otimes A)$, and contains block matrices~$W_{ji} A$ where~$W_{ji} \neq 0$ is the element in $j$th~row and $i$th~column of the full matrix $W$. It follows that
\begin{equation}
S\mbox{-rank}\left(
\left[
\begin{array}{c}
W_{ji}A\\
C^T_iC_i
\end{array}
\right]
\right) <n,
\end{equation}
for all~$j=1,...,N$ as~$w_{ji}\neq0$ and scalar multiplication does not change the structure and the $S$-rank (maximum possible rank over all values). Since~$A$ is not full~$S$-rank, $W_i \otimes A$ has rank less than~$n$ as stacking matrices with the same structure on top of each other (see Fig.\ref{figWAC}-Left) does not improve the~$S$-rank, which immediately results in
\begin{equation}
S\mbox{-rank}\left(
\left[
\begin{array}{c}
W_{i} \otimes A\\
C^T_iC_i
\end{array}
\right]
\right) <n.
\end{equation}
Consequently, according to Fig.\ref{figWAC}-right, the structure of the matrix $W \otimes A$ is given as the side-by-side concatenation of the matrices $W_i \otimes A$. Thus we have,
\begin{figure}
\centering
\includegraphics[width=3 in]{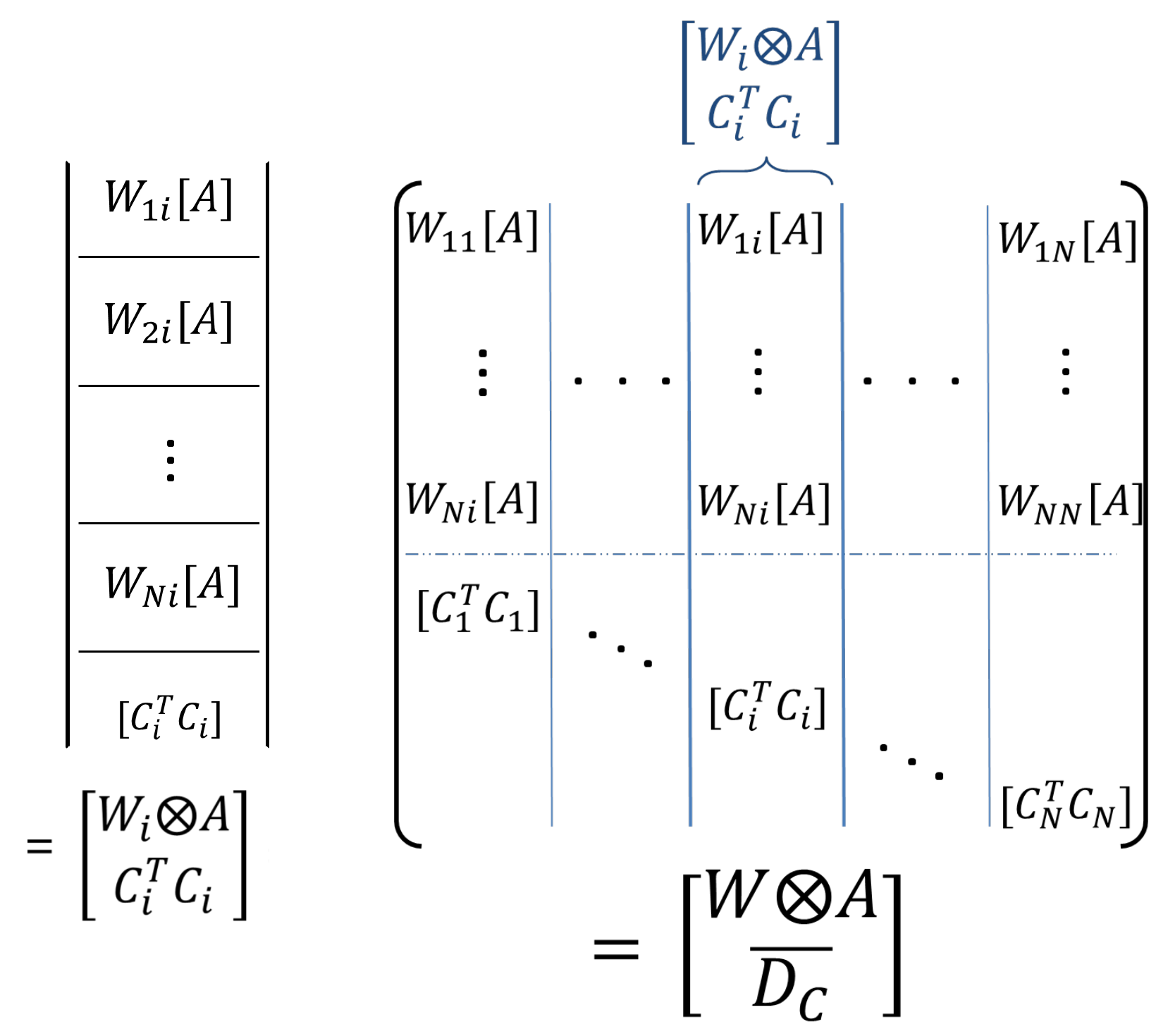}
\caption{The figures illustrates the structure of $W_i \otimes A$ (left) and  matrix $W \otimes A$ (right) in the proof of Theorem~\ref{non-full A,W}.}
\label{figWAC}
\end{figure}
\begin{equation}
S\mbox{-rank}\left(
\left[
\begin{array}{c}
W \otimes A\\
\overline{D_C}
\end{array}
\right]
\right) <Nn.
\end{equation}
This holds for almost all choices of non-zero elements in the $W$~matrices. Clearly, for any lower $S$-rank~$W$ the rank of $W \otimes A$ cannot be recovered as well. Therefore, according to Lemma~\ref{lemth1}, the condition (i) in Theorem~\ref{th1} is violated.
\end{proof}
The above theorem shows that when $A$ is $S$-rank deficient, then using state fusion cannot always guarantee the observability of the system and thus, the agents need access to more measurement data to recover observability, which is discussed next. In contrast, Theorem~\ref{fullA,W} shows that when $A$ is full $S$-rank, then state fusion alone can recover observability and further provides a method for the required agent communication. Clearly, these two results are to be viewed as a direct consequence of Assumptions (ii) and (iii) in Section~\ref{as_cl_a}. To the best of our knowledge, these conditions have not been developed before.

\subsection{Output fusion}
The other solution to make the NKE observable is output fusion, that is the second update level given in the equation~\eqref{le}. According to the formulation, each agent shares its measurement with its direct neighbors and implements this as an innovation to update its prediction. According to Table~\ref{tab}, for output fusion the structure of the matrix~$D_C$ has to be determined such that $I \otimes A,D_C$ is observable.

Based on the definition of~$D_C$, the~$i$th $n \times n$ diagonal block of~$D_C$ contains all of the measurements in the extended neighborhood of agent~$i$. In the distributed system graph ~$G_{(I \otimes A,D_C)}$, say for the agent~$i$, this is equivalent to adding all measurements in the neighborhood~$\mathcal{N}_i$ to the subsystem~$i$.

We now provide our main result on output fusion.
\begin{theorem} \label{Dc}
The system~$(I \otimes A,D_C)$ is observable if and only if:
\begin{enumerate}[(i)]
\item The sub-graph of all Type-$\alpha$ and Type-$\beta$ agents is a complete graph, i.e. all these crucial agents are needed to be directly linked together;
\item Every Type-$\gamma$ agent is directly connected to all Type-$\alpha$ and Type-$\beta$ agents.
\end{enumerate}
\end{theorem}
\begin{proof}

\emph{Sufficiency:} With the given conditions (i) and (ii), each agent has access to all crucial measurements. This makes every agent generically observable.

\emph{Necessity:} If an agent is not connected to one of the crucial agents, then it is missing a crucial measurement and the statement follows.
\end{proof}
It can be verified that if a system is not~$(A,C)$ observable then even
using a fully-connected communication network does not recover
observability. Clearly, the only way to get a stable estimation error is by increasing the number of state observations and recovering the $(A,C)$~observability. An interesting result on how to recover $(A,C)$~observability can be found in~\cite{boukhobza-recovery}.

\subsection{Main result}
Finally, we consolidate our results in previous subsections on state and output fusion. Theorem~\ref{fullA,W} sets the condition for state fusion for full $S$-rank systems, i.e., conditions for $(W \otimes A,\overline{D_C})$~generic observability, while Theorem~\ref{non-full A,W} states that for general $S$-rank deficient systems networked observability cannot be achieved via the state fusion alone. Output fusion, i.e., generic observability of $(I \otimes A,D_C)$, is discussed in Theorem~\ref{Dc}. Combining these results, we now provide the main theorem on generic observability of the single time-scale NKE protocol in Eqs.~\eqref{lp}--\eqref{le}.

\begin{theorem} \label{mainth}
For $(W \otimes A, D_C)$~observability with minimal number of communications, each agent needs:
\begin{enumerate}
\item A direct link \textit{from} all the Type-$\alpha$ agents (output-fusion);
\item A directed path \textit{to} (at least) one Type-$\beta$ agent for every parent SCC of $A$. This means, if there is two or more agents observing the same parent SCC, a directed path to any one of them is sufficient (state-fusion).
\end{enumerate}
\end{theorem}
\begin{proof}
The proof is a direct consequence of Theorems~\ref{fullA,W},~\ref{non-full A,W}, and~\ref{Dc}.
\end{proof}

The following is a complementary remark to the Theorem~\ref{mainth}.
\begin{rem}
In the case of Type-$\beta$ agents, every agent requires either a directed path \textit{to} each Type-$\beta$ agent (as stated in the Theorem~\ref{mainth}) or a direct link \textit{from} each Type-$\beta$ agent (as stated in the Theorem~\ref{Dc}); either one of these two conditions is sufficient for observability. However, the first strategy requires less number of links compared to the latter one, and therefore, it is preferred in terms of the minimal number of links in the communication network.
\end{rem}

Notice that compared to the typical assumptions on the agents' network in the literature, like strong connectivity or having a cyclic path, here we provide milder condition on the non-crucial agents; as there is no need for connectivity \textit{to} these agents but \textit{from} these agents. Furthermore, an agent may have no measurement of the system and still be able to estimate the state of the system via the proposed strategies. Such agents, for example, may play a role to provide and maintain the connectivity of the agent communication network \cite{Kopeikin11}, or even, maintaining directed paths to Type-$\beta$ agents as stated in the second condition of the Theorem~\ref{mainth}.

\section{Design of local estimator gain} \label{K}
In this subsection, we consider the design of the estimator gain matrix, $K_k$. Notice that having $(W \otimes A, D_C)$~observable guarantees a \textit{full} gain matrix, $K_k$, to stabilize the NKE error. However, according to protocol \eqref{le}, we need a local gain matrix, $K_k$, which is \textit{block-diagonal} with $N$ blocks of $n\times n$ matrices. For this section, we assume a constant estimator gain matrix is applied, i.e., the matrix $K_k$ is independent of time, $k$, and denote it by $K$.

A partial list of references devoted to find constrained estimator gain for control and estimation is~\cite{usman_cdc:11,KhanJadArxive,5717159,rami:97,siljak08}.
Here, we use the Linear Matrix Inequality (LMI) approach in \cite{5717159,rami:97}. However, in general, the corresponding LMIs do not have a solution, because of the structural constraints (block-diagonal) on the gain matrix, $K$. This is the main difficulty in distributed estimation and control as convex/semidefinite approaches are not directly applicable. To this end, we implement an iterative procedure to solve LMIs under
structural constraints.

In this regard, the following lemma presents the optimization approach to solve the estimator gain design problem. Interested readers may find more details in \cite{usman_cdc:11,5717159}.
\begin{lem} \label{optm}
If the NKE protocol \eqref{lp}--\eqref{le} is observable, then estimator gain matrix, $K$, is the solution of the following optimization problem.
\begin{equation}
\begin{aligned}
& \displaystyle\min ~~  \mathbf{trace}(XY) \\
& \text{subject to} ~~ X,Y>0,\\
& ~~~~~~~~~~~~ \left[ \begin{array}{cc} X&\widehat{A}^T\\ \widehat{A}&Y\\ \end{array} \right] > 0,
& \left[ \begin{array}{cc} X&I\\ I&Y\\ \end{array} \right]>0,\\
& ~~~~~~~~~~~~~ K\mbox{~is~block-diagonal}.
\end{aligned}
\label{min}
\end{equation}
where,
\begin{equation}
\widehat{A} = W \otimes A - KD_C(W \otimes A)
\end{equation}
\end{lem}
In fact, we need a block-diagonal $K$ such that $\widehat{A}$ is Schur (i.e. $\rho (\widehat{A}) <1$).

Notice that, the solution to the second LMI is equivalent to $X = Y^{-1}$, which gives the minimum trace and the optimal value as $nN$. The nonlinear product of X and Y can be replaced with a linear approximation \cite{5717159,rami:97,pang:95},
$\phi_{lin}(X,Y) = \mathbf{trace}(Y_0X + X_0S)$
and an iterative algorithm \cite{rami:97} can be used to minimize $\mathbf{trace}(XY)$
under the given constraints.
\begin{algorithm} 
\caption{Iterative calculation of local gain estimator, $K$.}
\begin{algorithmic}
\item Find feasible points $X0, Y0,K$. If no such points exist,
Terminate.
\item At iteration $t>0$ minimize $\mathbf{trace}(Y_tX + X_tY)$ under
the constraints given in \eqref{min} and find $X,Y,K$.
\item If $\rho (\widehat{A}) <1$ terminate, otherwise set $Y_{t+1}=Y,~X_{t+1}=X$ and run the step 2 for next iteration $t=t+1$.
\end{algorithmic}
\label{algorithm}
\end{algorithm}

It is shown in \cite{rami:97} that $\mathbf{trace}(Y_tX + X_tY)$ is a non-increasing sequence that converges to $2nN$. In this regard, a stopping
criterion in step 3 of the above algorithm can also
be established in terms of reaching within $2nN + \epsilon$ of the
trace objective. The iterative procedure given above
is centralized, however, the center has to implement this process
only once, off-line; then it transmits the
estimator gains to each agent and plays
no further role in the implementation of local estimators; each agent, subsequently, observes and performs in-network operations to implement the estimator.
A single time-scale algorithm can also be employed,
where the above iterative procedure is implemented at the
same time-scale $k$ as of the dynamical system. With
this approach, the estimator gain iterations, $K_{t}$, at each $t$ is applied to the estimator at time-step, $k$, and may be transmitted to each agent at each step $k$. This is helpful when the implementation is assumed in real-time.

\section{Illustrative example and simulation}\label{example}
Consider the system,~$(A,C)$, given in Fig.~\ref{figAC}. The structure of these matrices is given by
\begin{eqnarray}\label{Asim}
A = \left[
\begin{array}{ccccccc}
0&\times&0&0&0&0&0\\
\times&0&0&0&0&0&0\\
0&\times&0&0&0&0&0\\
0&0&\times&0&\times&\times&0\\
0&0&0&\times&0&0&0\\
0&0&0&\times&\times&0&\times\\
0&0&0&0&0&0&\times\\
\end{array}
\right],
\end{eqnarray}
\begin{eqnarray}\label{Csim1}
\left[
\begin{array}{c}
C_a\\
C_b\\
C_c
\end{array}
\right]&=&\left[
\begin{array}{ccccccc}
0&0&\times&0&0&0&0\\
0&0&0&0&\times&0&0\\
0&0&0&0&0&0&\times
\end{array}
\right].
\end{eqnarray}
Now, recall that based on Theorem~\ref{th1}, the system is globally observable by collection of the three measurements. The state partitioning with maximal cycles so that this partitioning covers all the states (in light of condition (i)--Theorem~\ref{th1}) is $\mathcal{L}=\{(4,5,6,4), (1,2,1), (7,7), (3,a)\}$. By definition, agent~$a$ is Type-$\alpha$, agent~$b$ is Type-$\beta$ and agent~$c$ is Type-$\gamma$. It can further be verified that agents $a$ and $b$ possess crucial observations. 

To better illustrate the networking effect, we first note that the
networked system with no information sharing, i.e., the graph
associated to~$(I \otimes A, \overline{D_c})$, is not observable at any of the agents individually (Fig.~\ref{figA}). To make the networked system observable at each agent, we propose the following communication matrices~$W_1$ and~$W_2$, and their associated graphs~$G_W$ in Fig.~\ref{fignet} as two minimal networks making the system observable.

\begin{eqnarray}\label{CsimCE}
W_1 = \left[
\begin{array}{ccc}
\times&\times&0\\
\times&\times&0\\
\times&\times&\times
\end{array}
\right],\qquad
W_2 =\left[
\begin{array}{ccc}
\times&0&\times\\
\times&\times&0\\
\times&0&\times
\end{array}
\right].
\end{eqnarray}

The graph associated to~$W_1$ is proposed based on the
Theorem~\ref{Dc}. In this network crucial agents,~$\{a,b\}$,
are directly linked among each other and both have a directed link to agent, $c$, with no crucial observation. The second communication network~$W_2$ is based on Theorem~\ref{mainth}; there is a direct link from agent~$a$ (Type-$\alpha$) to all other agents, and there is a path from every other agent to agent~$b$ (Type-$\beta$). It can be verified that for both topologies $(W \otimes A,D_C)$ is generically observable.
\begin{figure}
\centering
\includegraphics[width=3in]{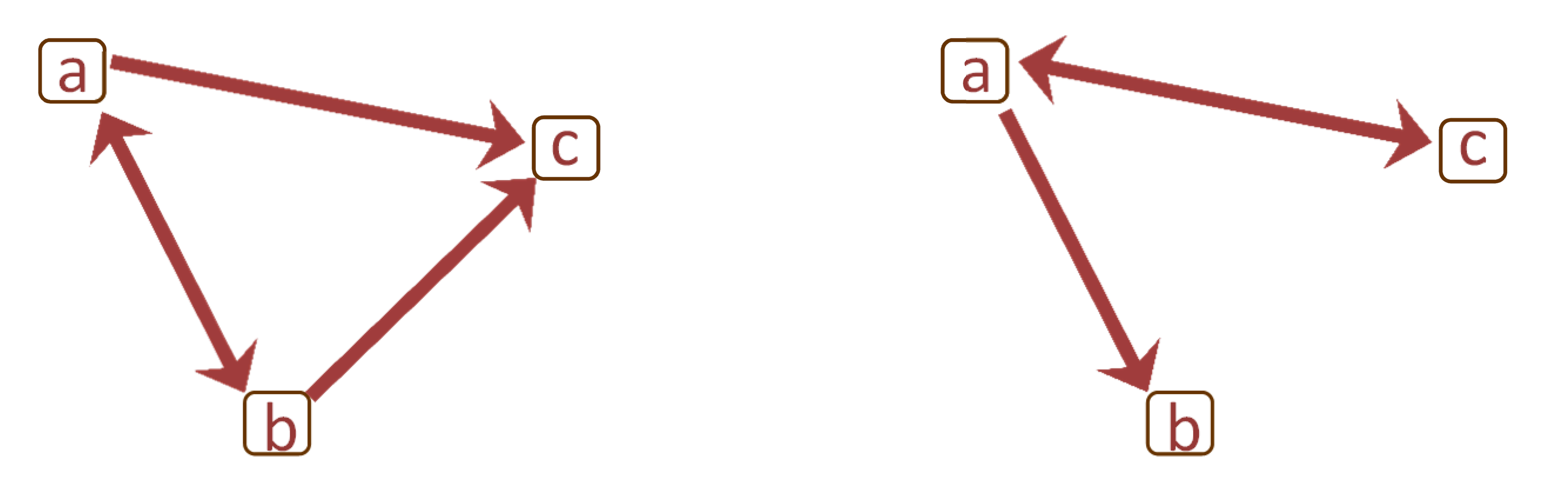}
\caption{(Left) communication network for the three agents using
 output fusion (Theorem~\ref{Dc}), (Right) both measurement and state fusion (Theorem~\ref{mainth}).}
 \label{fignet}
\end{figure}
Note that the solution for the network design problem is not unique,
and there maybe other examples of communication network satisfying
the conditions in the last section. In addition, any network
including one of these two topologies as a sub-graph is also a
solution to the NKE problem. For example, under
the full structured rank assumption of system~$A$, any strongly
connected network among agents suffices for individual
observability~\cite{usman_cdc:11}.

For simulation, we consider a random valued matrix, $A$, with the structure in Eq.~\eqref{Asim}, and an output matrix, Eq.~\eqref{Csim1}, with all non-zero entries equal to $1$. The system eigenvalues are as follows.
\begin{eqnarray*}
eig(A) = [-1.0838,1.0838 ,0.6511, -0.5571,\ldots\\ -0.0940, 0.0000,1.3072]
\end{eqnarray*}
Clearly, the dynamical system is unstable, since $\rho (A)= 1.0838$. We choose the agents' network according to Fig.~\ref{fignet} (Right) with random link weights such that it remains stochastic. We use Algorithm~\ref{algorithm} to find the block-diagonal gain matrix, $K$. Using this gain matrix, the eigen-values of the error dynamics, i.e., of the matrix, $\widehat{A}$, are as follows,
\begin{eqnarray*}
eig(\widehat{A}) = [0.8190 , 0.6511 ,-0.5571, 0.0073 \pm 0.3159i,\ldots\\
-0.2700,0.1643, -0.1406, -0.0940,\\
-0.0788, 0.0214, -0.0237,0, 0,0, 0,0, 0 ,0, 0,0],
\end{eqnarray*}
which are all stable. The system and output noise are, respectively, $\mb{v}_k\sim\mathcal{N}(0,0.05)$ and ~$\mb{r}^i_k\sim\mathcal{N}(0,0.2)$. As system initial state we choose a random initial value between $0$ and $3$. The system error evolution over $100$ iterations for three agents are given in Fig.~\ref{figsim}. For the visual clarity, we have squared the errors at each iteration and then summed them over the $n=7$ states of the dynamical system. 
\begin{figure}
\centering
\includegraphics[width=3in]{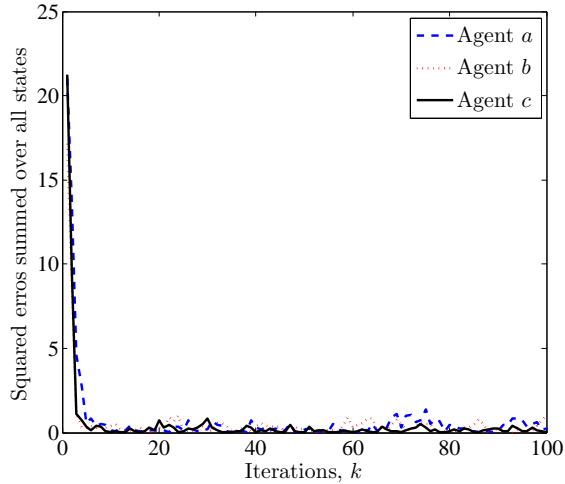}
\caption{Performance of the networked estimator at each agent. The error is squared and then summed over $n=7$ states at each iteration.}
\label{figsim}
\end{figure}
As it can be seen, the estimation error at all agents is bounded despite the fact that system is not stable.
\section{Conclusion}\label{conc}
In this paper, we study the role of the agent communication
network towards error stability of the NKE protocol~\eqref{lp}--\eqref{le} in the context of single time-scale distributed estimation. As opposed to multiple time-scale strategies where the communication network is irrelevant and diffusion strategies where the estimator error is \textit{irrespective} of system dynamics, here, we take into account both system dynamics and communication network. We show that the NKE is able to track even potentially unstable dynamical systems, i.e., the networked estimator is observable for all stable and unstable eigenvalues. We show that under sufficient communication among the agents, the system state is generically observable at every agent. Here, we provide minimal sufficient network connectivity applicable for multi-agent systems subject to constrained communication, e.g., out-of-range geographical distances or costly communication. We define three types of agents/measurement where two types are crucial for observability. We provide two main results on recovering networked observability: (i) with state-fusion, and (ii) with output-fusion. Furthermore, we determine dynamical systems ($S$-rank deficient) for which no state-fusion results in an observable networked estimator and one has to rely on output-fusion as well.

Our results are on the  existence of a network structure for bounded estimation error and further finding such network with minimal links. Because of the \emph{genericity}, the link weights are free parameters and results are independent of any particular fusion rule (e.g., Metropolis-Hastings~\cite{Xiao05distributedaverage}) chosen in~\eqref{lp}--\eqref{le}. Nevertheless, the structure of the underlying agent communication
remains relevant and leads to network/infrastructure design. Furthermore, link weights can be optimized to reduce the error, which is a direction for the future work. It is worth noting that, in general, $S$-rank and other generic properties are easily verified. For example, there are efficient graph theoretic,
\cite{rein_book}, flow theoretic, \cite{Hovelaque:1996:ALS:226451.226452}, and linear programming,\cite{poljak:1989}, methods that can be employed to check for generic properties.

\bibliographystyle{IEEEbib}
\bibliography{bibliography}

\end{document}